\documentclass{ifacconf}

\usepackage{natbib}        % required for bibliography
\usepackage{amsmath, amssymb, graphicx, subcaption}
\usepackage{nicefrac}
\usepackage{enumitem} % preamble에 추가
\usepackage{xcolor}

\newtheorem{theorem}{Theorem}

\newtheorem{remark}{Remark}
\newtheorem{lemma}{Lemma}
\newtheorem{proof}{Proof}

\begin{document}
\begin{frontmatter}

\title{AI-Augmented Density-Driven Optimal Control for Decentralized Environmental Mapping\thanksref{footnoteinfo}} 

\thanks[footnoteinfo]{This work was supported by NSF CAREER Grant CMMI-DCSD-2145810.}

\author[First]{Kooktae Lee} 
\author[First]{Julian Martinez}

\address[First]{Department of Mechanical Engineering, New Mexico Institute of Mining and Technology, Socorro, NM 87801, USA (email: kooktae.lee@nmt.edu / julian.martinez@student.nmt.edu).}
% \address[Second]{Department of Mechanical Engineering, New Mexico Institute of Mining and Technology, Socorro, NM 87801, USA (email: julian.martinez@student.nmt.edu).}

\begin{abstract}
This paper presents an AI-augmented decentralized framework for multi-agent environmental mapping under limited sensing and communication. 
While conventional coverage formulations achieve effective spatial allocation when an accurate reference map is available, their performance deteriorates under uncertain or biased priors. 
The proposed method introduces an adaptive and self-correcting mechanism that enables agents to iteratively refine local density estimates within an optimal transport–based framework, ensuring theoretical consistency and scalability.
A dual multilayer perceptron (MLP) module enhances adaptivity by inferring local mean–variance statistics and regulating virtual uncertainty for long-unvisited regions, mitigating stagnation around local minima. 
Theoretical analysis rigorously proves convergence under the Wasserstein metric, while simulation results demonstrate that the proposed AI-augmented Density-Driven Optimal Control consistently achieves robust and precise alignment with the ground-truth density, yielding substantially higher-fidelity reconstruction of complex multi-modal spatial distributions compared with conventional decentralized baselines.
\end{abstract}


\end{frontmatter}

%=================================================
\section{Introduction}

Accurate spatial mapping is essential for autonomous multi-agent systems in applications such as environmental monitoring, pollution tracking, precision agriculture, and disaster response~\cite{hollinger2014sampling,binney2013optimizing}. 
In practice, agents must operate under limited sensing, partial communication, and dynamic constraints, while prior environmental knowledge is often sparse or outdated, making it challenging to design decentralized strategies that achieve both efficiency and adaptivity.

Conventional non-uniform coverage frameworks, such as Density-Driven Control (D$^2$C)~\cite{afrazi2025enhanced} and its optimal-control extension, Density-Driven Optimal Control (D$^2$OC)~\cite{seo2025smartfarm,seo2025density}, achieve effective spatial allocation when an accurate reference map is available (see the conceptual illustration in Fig.~\ref{fig:d2oc}). 
However, their performance deteriorates when the prior map is uncertain or biased. 
Other representative approaches, including ergodic exploration~\cite{mathew2011metrics,gkouletsos2021decentralized}, information-theoretic coverage~\cite{julian2012distributed}, and distributed optimal transport methods~\cite{krishnan2025distributed,bandyopadhyay2017probabilistic}, also pursue spatially balanced exploration through information- or transport-theoretic formulations. 
Nevertheless, these methods typically rely on accurate or well-established priors or globally consistent cost functions, limiting their adaptability under partial sensing and intermittent communication. 
In realistic environments, however, only coarse and dynamically evolving estimates of the underlying field are available, motivating the need for an adaptive and self-correcting coverage mechanism capable of refining its reference model through continual interaction with noisy observations.

\begin{figure}[b!]
    \centering
    \subfloat[]{\includegraphics[width=0.475\linewidth]{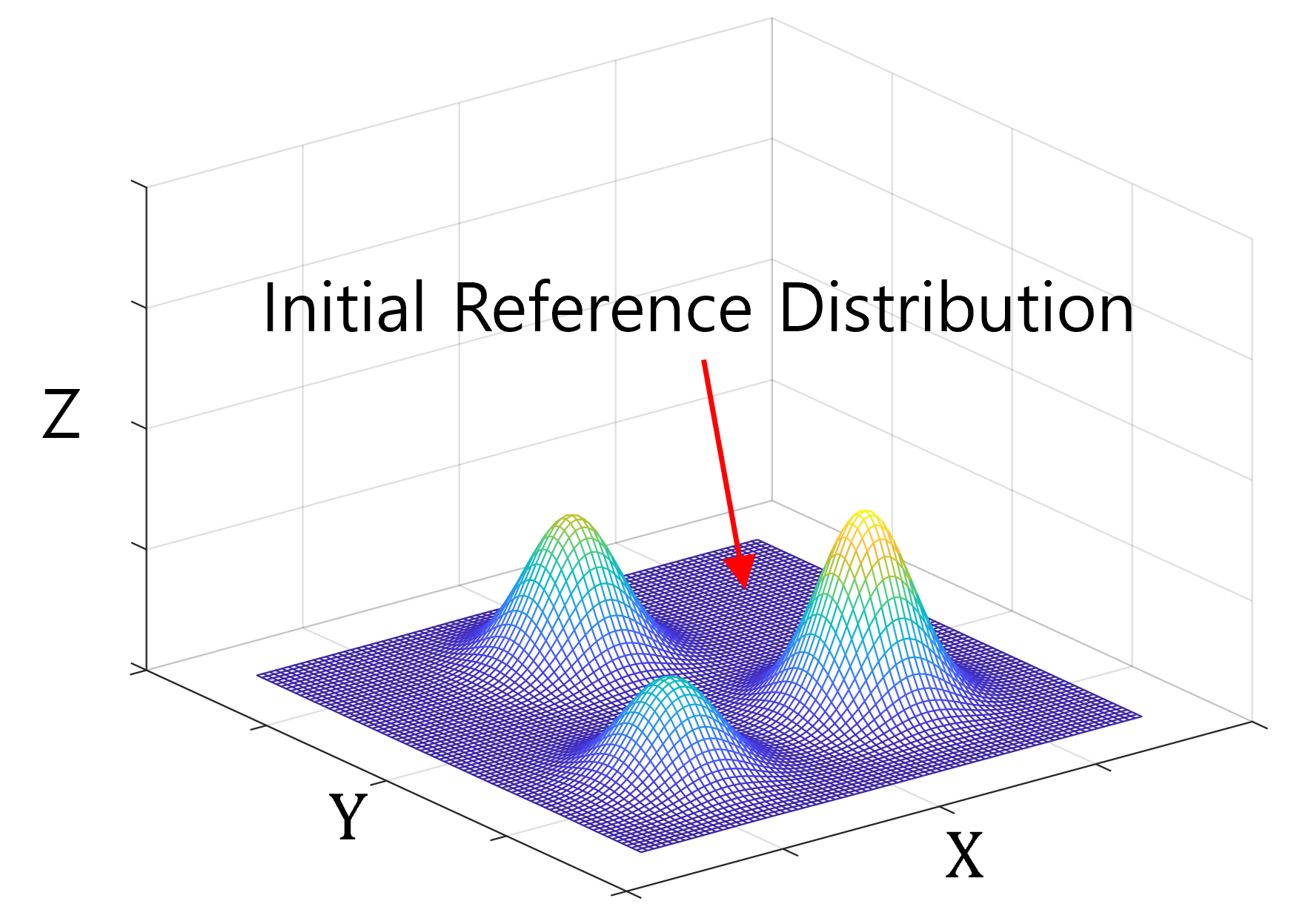}}
    \subfloat[]{\includegraphics[width=0.475\linewidth]{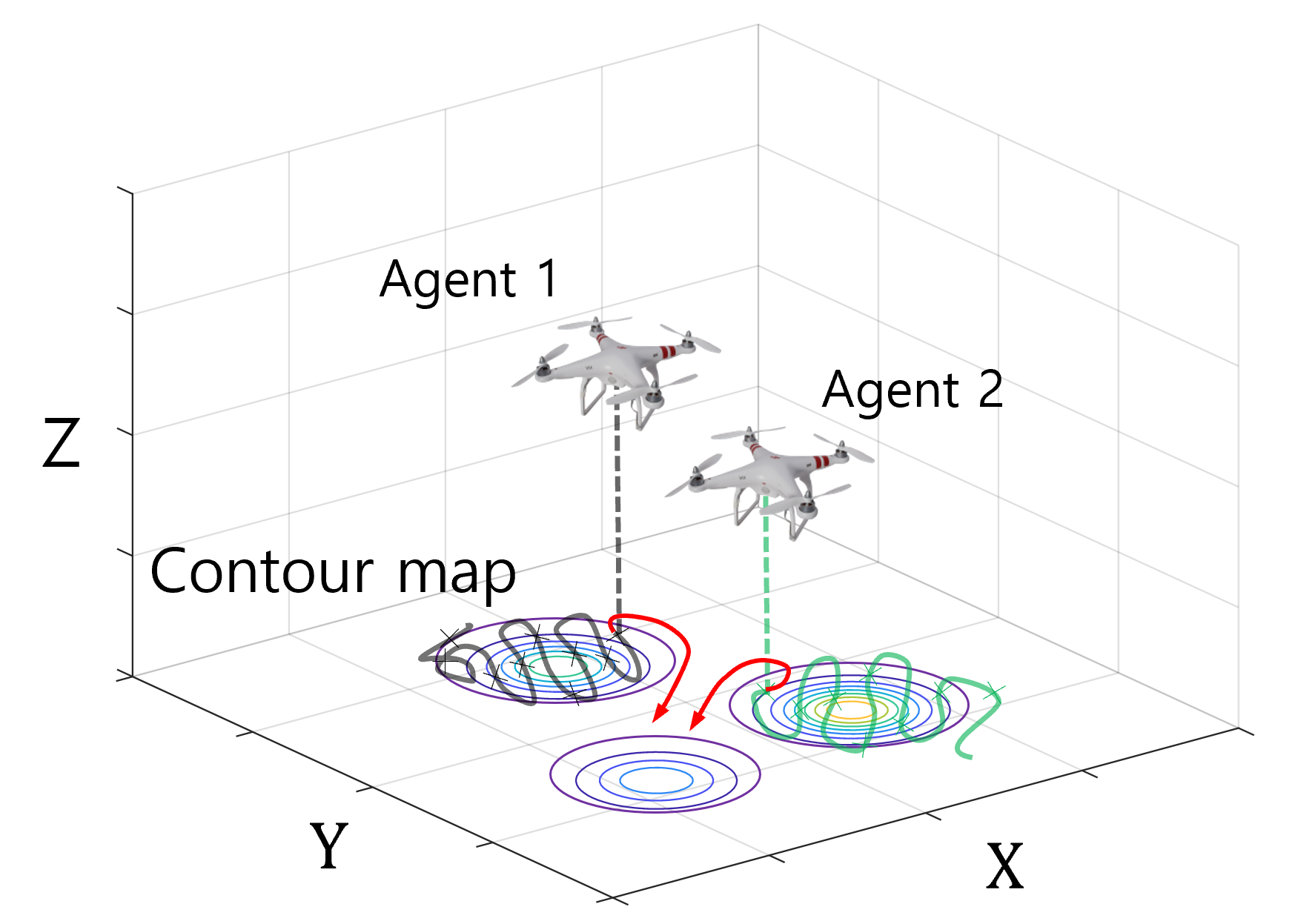}}
    \caption{Conceptual illustration of D$^2$C/D$^2$OC.}
    \label{fig:d2oc}
\end{figure}

To address these limitations, we propose an AI-augmented decentralized framework in which agents iteratively update local map estimates using online sensing and limited-range communication. 
This enables progressive reconstruction of the true spatial importance distribution while maintaining adaptive coverage. 
Unlike conventional informative path planning (IPP) schemes~\cite{binney2012branch,hollinger2014sampling}, which aim to maximize short-term information gain based on local uncertainty reduction, the proposed framework optimizes a distribution-level objective under the Wasserstein metric, ensuring globally consistent convergence toward the ground-truth density.

The proposed architecture builds upon the D$^2$OC formulation, which ensures mass-preserving updates under the Wasserstein metric for guaranteed convergence. 
It consists of three stages: 
(A) sample selection and optimal control, 
(B) coverage-tracking weight update, and 
(C) distributed consensus with range-limited communication. 
This structure supports fully decentralized implementation while maintaining theoretical consistency with optimal transport (OT) principles.

To enhance adaptivity, an AI module is incorporated into Stage~A. 
Each agent employs a dual multilayer perceptron (MLP) inference model: one network provides static mean–variance estimation, while the other dynamically adjusts virtual uncertainty for long-unvisited regions. 
These learned uncertainty weights modulate local sampling priorities, encouraging re-exploration of neglected areas without centralized coordination. 
The AI-augmented formulation preserves the analytical structure of D$^2$OC while improving robustness against sensing noise, sample renewal randomness, and local minima.

Unlike previous D$^2$C and D$^2$OC studies that assume a known and accurate reference map, this work addresses the fundamentally different problem of decentralized map reconstruction under uncertain or incomplete priors. 
The main contributions are summarized as follows:
1) A decentralized framework that reconstructs initially uncertain reference maps through local sensing and neighbor communication without central coordination;
2) Integration of an AI-based dual-MLP inference module into D$^2$OC to guide exploration based on predicted importance and uncertainty;
3) Theoretical proof of convergence under the Wasserstein metric, supported by simulations showing that the AI-augmented D$^2$OC avoids stagnation under biased priors and achieves a significantly lower steady-state Wasserstein distance than the non-AI baseline.

%=================================================

\section{Problem Description and Objectives}
\noindent
\textbf{Notations:}
Let $\mathbb{R}$ denote the set of real numbers, and $\mathbb{N}_0 := \{0,1,2,\ldots\}$ the set of nonnegative integers.
For any vector or matrix, $(\cdot)^\top$ denotes transpose.
The discrete-time index is indicated by the superscript $k \in \mathbb{N}_0$, e.g., $\mathbf{x}^k$ represents the state at time step $k$.
The Euclidean norm is denoted by $\|\cdot\|$, and the weighted norm is defined as $\|x\|_R^2 = x^\top R x$ for a positive-definite matrix $R \succ 0$.
The Dirac delta function is denoted by $\delta(\cdot)$.
For any set $\mathcal{S}$, $|\mathcal{S}|$ denotes its cardinality (i.e., the number of elements in $\mathcal{S}$).

We consider a team of $N_{\mathrm{a}}$ mobile agents operating in a bounded domain $\Omega \subset \mathbb{R}^d$.  
The true spatial importance of each point $x \in \Omega$ is represented by an unknown distribution $\rho_{\mathrm{GT}}$, indicating the priority of monitoring that region.  
In practice, agents initially possess only a coarse prior map $\hat{\rho}$ from historical or low-fidelity data, and must therefore \emph{explore the environment while refining their local estimates} to approximate $\rho_{\mathrm{GT}}$.

Each agent $i$ maintains a local \emph{sample-based reference map}:
\begin{equation}
\textstyle
\hat{\rho}_i^k(x) = \sum_{j=1}^{N_{\mathrm{s}}} w_{j}^k \, \delta(x - q_{j}^k),
\end{equation}
where $N_{\mathrm{s}}$ is the number of sample points, $q_{j}^k \in \Omega$ denotes the $j$th sample location, and $w_{j}^k$ its associated weight. 
Weights represent local importance inferred from sensing and are updated via communication with nearby agents.  
All computations are fully decentralized without any global fusion.

To ensure energy-aware exploration, each agent limits its operation to a finite set of discrete points~\(M\),
determined by its operation time \(T_{\mathrm{op}}\) and sampling interval \(\Delta t\),
i.e., \(M = T_{\mathrm{op}} / \Delta t\).
Each agent point is associated with a normalized mass \(1/M\),
and the sample weights satisfy \(\sum_j w_j^k = 1\),
defining the empirical distribution \(\hat{\rho}_i^k\) that evolves as exploration proceeds.

The objective is to minimize the discrepancy between each agent’s local map~$\hat{\rho}_i^k(x)$ and the ground-truth density~$\rho_{\mathrm{GT}}(x)$,
measured by the 2-Wasserstein distance:
\begin{equation}
\mathcal{W}_2^2\big(\hat{\rho}_i^k, \rho_{\mathrm{GT}}\big)
= \min_{\pi \in \Pi(\hat{\rho}_i^k,\rho_{\mathrm{GT}})}
\sum_{i,j} \pi_{ij}\, \| p_i - q_j \|^2, \label{eq:W_def}
\end{equation}
where $\Pi(\hat{\rho}_i^k,\rho_{\mathrm{GT}})$ denotes feasible transport plans
$\pi_{ij}\!\ge\!0$ satisfying
$\sum_j \pi_{ij} = 1/M$ and $\sum_i \pi_{ij} = w_j$.

We aim to design decentralized update rules and control inputs such that
\begin{equation}
\lim_{k \to \infty} \mathcal{W}_2\big(\hat{\rho}_i^k, \rho_{\text{GT}}\big) \le \epsilon, \quad \forall i,
\end{equation}
where $\epsilon$ is a small residual error due to finite sampling.

While classical OT seeks the optimal transport plan $\pi_{ij}$ minimizing \eqref{eq:W_def} for fixed point sets $\{p_i\}$ and $\{q_j\}$,
the proposed D$^2$OC framework instead determines the \emph{control-driven sequence of $\{p_i^k\}$} that gradually reduces this distance, which is an inherently much more challenging problem.

Conceptually, this involves two coupled objectives:
\begin{enumerate}[leftmargin=*, itemsep=1pt]
    \item \textbf{Adaptive coverage:} Agents move so that their collective trajectories are proportional to the estimated importance distribution, prioritizing high-value regions.  
    \item \textbf{Incremental map refinement:} As agents sense new data, they update the sample weights $w_{j}^k$, improving local map fidelity and guiding subsequent motions.
\end{enumerate}

The coupling can be formalized by the time-averaged empirical measure:
\begin{equation}
\textstyle\hat{\rho}_{i,\textrm{agent}}^K(x) = \frac{1}{K} \sum_{k=1}^{K} \delta(x - p_i^{k}),
\end{equation}
where $p_i^{k}$ is the agent’s position at step $k$.  
Accurate mapping is achieved when $\hat{\rho}_{i,\textrm{agent}}^K(x)$ matches $\hat{\rho}_i^k(x)$, which in turn converges toward $\rho_{\mathrm{GT}}$ via iterative weight and sample updates.  
Thus, motion and weight adaptation jointly minimize the Wasserstein distance between the empirical $\hat{\rho}$ and true distributions $\rho_{\mathrm{GT}}$.

In practice, the decentralized D$^2$OC framework~\cite{seo2025density} proceeds in three stages (A–C):
agents compute local centroids for motion, update sample weights, and exchange information with neighbors.
Unlike the original D$^2$OC formulation, which assumes that the reference distribution is perfectly known \emph{a priori} and does not require sensing-based updates,
the proposed AI-augmented framework newly incorporates an inference mechanism that estimates uncertainty in unobserved regions and prioritizes sampling.
This extension enables D$^2$OC to operate under realistic conditions where the true distribution is initially unknown and must be progressively inferred through exploration.

The overall goal is thus to design \emph{decentralized iterative rules for motion and weight updates} so that each local map $\hat{\rho}_i^k(x)$ converges to $\rho_{\mathrm{GT}}$ in Wasserstein distance, while collective trajectories realize non-uniform coverage aligned with the evolving map.
This formulation motivates the subsequent development of the control mechanism, AI inference, and convergence analysis.

%=================================================
\section{Decentralized Density-Driven Optimal Control}

This section formulates the proposed AI-augmented D$^2$OC framework for decentralized multi-agent coverage in environments with unknown or imperfectly estimated spatial importance distributions.
 
Before executing control updates, each agent renews its reference sample set through a \textit{Dynamic Sample Generation and Removal} process.  
The main procedure then consists of three stages: 
\textit{Stage~A: Sample Selection \& Optimal Control}, 
\textit{Stage~B: Weight Update}, and 
\textit{Stage~C: Communication \& Merging}.  
An AI-based module estimates sample uncertainty for computing priority scores, as detailed in Section~\ref{sec:AI}.

\subsection{Dynamic Sample Generation and Removal}

Before the control update in Stage~A, each agent renews its sample set 
$\mathcal{S}_{i}^k = \{ q_{j}^{k} \}$ based on recent sensory measurements.  
At each step, agent~$i$ explores its neighborhood within a sensing range~$r_\mathrm{s}$ 
and generates candidate samples according to the measured field intensity.
Each candidate location $x_j^k$ within the sensing range is evaluated through a noisy sensing model
\[
\xi_{j}^{k} = \rho_{\mathrm{GT}}(x_{j}^{k}) + \varepsilon_{j}^{k},
\qquad \varepsilon_{j}^{k} \sim \mathcal{N}(0,\sigma_\mathrm{sensor}^2),
\]
where $\xi_{j}^{k}$ denotes the instantaneous sensor reading. 
From accumulated measurements and neighborhood statistics, 
the local mean $\mu_j^{k}$ and variance $(\sigma_j^{k})^2$ are estimated, 
representing the agent’s belief about the field intensity at $x_j^k$.  
Candidates with higher measured values are retained and added as new samples $q_j^k$ to the reference set.  
Accepted points are initialized with nominal weights and variances, forming the newly born subset of $\mathcal{S}_i^k$.

Meanwhile, existing samples within the sensing region are examined for removal.  
If their field intensity falls below a distance-dependent threshold, they are discarded, representing the death process.  
This \textit{birth–death} mechanism continuously adapts the reference distribution~$\{q_j\}$ to new sensory information, maintaining a concise and up-to-date sample representation for subsequent control-driven redistribution in Stage~A.

\subsection{Stage A: Sample Selection and Optimal Control}
Let \(p_i^{k} \in \mathbb{R}^d\) denote the position of agent \(i\) at time \(k\).
Each sample \(q_{j}^k \in \mathcal{S}_i^k\) is assigned an \textit{importance score}
\begin{equation}
\phi_j^{k} = \mu_j^{k} + c_1 (\sigma_j^{k})^2 + c_2 (\sigma_{j,\mathrm{virtual}}^{k})^2,
\label{eq:phi}
\end{equation}
where $c_1, c_2$ are some positive coefficients, \(\mu_j^{k}\) and \((\sigma_j^{k})^2\) are the locally measured mean and variance 
obtained from sensor feedback at step \(k\),
and \((\sigma_{j,\mathrm{virtual}}^{k})^2\) is an auxiliary variance adaptively estimated 
by a separate AI module to encourage exploration of long-unvisited regions.
Two distinct neural modules are employed in practice:
one refines the measured statistics \((\mu_j^{k}, \sigma_j^{k})\) to suppress sensor noise,
while the other updates \(\sigma_{j,\mathrm{virtual}}^{k}\) online based on the agent’s visitation history.
The combined score balances reliable exploitation of informative regions 
(\(\mu_j^{k}, \sigma_j^{k}\)) with active exploration driven by the virtual uncertainty term.
This importance score is then used in the weighted centroid computation of Stage~A
to determine each agent’s next reference position.

The \textit{score} of each sample balances importance and proximity $s_j^k = \frac{\phi_j^k}{\|p_i^{k} - q_{j}^k\|}$.
Samples are sorted in descending order of \(s_j\),
and selected sequentially until the cumulative weight
reaches the agent’s total mass budget \(1/M\): $
\textstyle \sum_{j \in \mathcal{S}_{i,\mathrm{loc}}^{k}} w_j = \frac{1}{M}$.
The resulting subset \(\mathcal{S}_{i,\mathrm{loc}}^{k}\)
represents the locally influential samples that participate in the optimal control update for agent~\(i\).

\noindent\textbf{Optimal Control Formulation:}
Each agent follows discrete-time LTI dynamics:
\begin{equation}
\begin{aligned}
\mathbf{x}_i^{k+1} = A_i \mathbf{x}_i^{k} + B_i u_i^{k}, \quad
p_i^{k} &= C_i \mathbf{x}_i^{k},\label{eq:LTI}
\end{aligned}
\end{equation}
where $A_i \in \mathbb{R}^{n \times n}$, $B_i \in \mathbb{R}^{n \times m}$,
and $C_i \in \mathbb{R}^{d \times n}$ are the system matrices of agent~$i$,
and $u_i^{k} \in \mathbb{R}^m$ denotes the control input at time~$k$.

The finite-horizon optimal control problem minimizes the cumulative transport and control effort over a planning horizon~$H$:
\begin{equation}
\begin{aligned}
    J_i(u_i^{k:k+H-1})
    &= \sum_{l=k}^{k+H} \sum_{j \in \mathcal{S}_{i,\mathrm{loc}}^{l}}
       \pi_{lj}^*\, \| p_i^{l+1} - q_j^k \|^2
       + \sum_{l=k}^{k+H} \| u_i^{l} \|_R^2, 
       \label{eq:cost_func}
\end{aligned}
\end{equation}
subject to the discrete-time LTI dynamics in~\eqref{eq:LTI},  
where $H$ is the prediction horizon, $\pi_{lj}^*$ denotes the optimal transport plan with an analytic solution available in~\cite{kabir2021wildlife},  
$\mathcal{S}_{i,\mathrm{loc}}^{l}$ is the local sample set at step~$l$,  
and $R \succ 0$ is the control weighting matrix. Here, the first term denotes the local Wasserstein distance and the second term reflects the input penalty.

\noindent\textbf{Optimal Control via Weighted Centroid}:
For each agent \(i\), let \(\mathcal{S}_{i,\mathrm{loc}}^{l}\) denote the set of selected samples at discrete time \(l\) and \(\pi_{lj}^*\) be the transport weight associated with each sample \(q_{j}^k \in \mathcal{S}_{i,\mathrm{loc}}^{l}\). Define the weighted centroid of the selected samples for agent $i$ at time \(l\) as
\begin{equation}
\textstyle
q_{i,c}^{l} = 
\sum_{j \in \mathcal{S}_{i,\mathrm{loc}}^{l}} \pi_{lj}^* q_j^l \big/ \sum_{j \in \mathcal{S}_{i,\mathrm{loc}}^{l}} \pi_{lj}^*\label{eq:centroid}
\end{equation}

Then, the following result holds.

\begin{lemma}[Weighted centroid minimizes the local quadratic cost]
For agent $i$ at step $l$, consider 
\[
\textstyle
J_{i,\mathrm{dist}}^{l} = \sum_{j \in \mathcal{S}_{i,\mathrm{loc}}^{l}} \pi_{lj}^*\|p_i^{l+1} - q_{j}^k\|^2, 
\]
where $\pi_{lj}^* \ge 0$ denote the optimal transport coefficients. 
Let $\gamma_l := \sum_{j \in \mathcal{S}_{i,\mathrm{loc}}^{l}} \pi_{lj}^*$ and 
$q_{i,c}^{l} := \frac{1}{\gamma_l} \sum_{j \in \mathcal{S}_{i,\mathrm{loc}}^{l}} \pi_{lj}^* q_{j}^k$. 
Then $J_{i,\mathrm{dist}}^{l}$ is minimized by $p_i^{l+1} = q_{i,c}^{l}$.
\end{lemma}

\begin{proof}
Using the standard identity for weighted squared distances,
$\sum_{j} \pi_{lj}^* \|p_i^{l+1} - q_{j}^k\|^2
= \gamma_l \|p_i^{l+1} - q_{i,c}^{l}\|^2
+ \sum_{j} \pi_{lj}^* \|q_{j}^k - q_{i,c}^{l}\|^2$,
where $\gamma_l = \sum_j \pi_{lj}^*$.
The second term is independent of $p_i^{l+1}$,
so the minimum is attained at $p_i^{l+1} = q_{i,c}^{l}$.
Since the centroid is recomputed at every control step,
this property holds sequentially across the $H$-step horizon,
ensuring local optimality at each stage.
\end{proof}

\noindent\textbf{Optimal Control under Discrete-Time LTI Dynamics}:
The previous lemma characterizes the instantaneous optimal position update that minimizes the local quadratic cost.  
In a more realistic setting, however, agent motion is constrained by discrete-time linear dynamics.  
We therefore extend this result to derive a closed-form control law that optimally drives each agent toward its weighted centroid over a finite prediction horizon.

\medskip
\begin{theorem}[Analytic Receding-Horizon Control Law in D$^2$OC]\label{thm:opt_con}
Consider agent $i$ with discrete-time LTI dynamics \eqref{eq:LTI}
and the stage cost induced by \eqref{eq:cost_func} at step $l$.
Let the weighted centroid of $\mathcal{S}_{i,\mathrm{loc}}^{l}$ be defined in \eqref{eq:centroid} with
$\gamma_l := \sum_{j \in \mathcal{S}_{i,\mathrm{loc}}^{l}} \pi_{lj}^*$.
Then the control input that minimizes the stage cost at step $l$ is
\begin{equation}
    \begin{aligned}
        u_i^{\,l}
        &= (R + \gamma_l B_i^\top B_i)^{-1}
          \gamma_l B_i^\top (q_{i,c}^{\,l} - A_i \mathbf{x}_i^{\,l}),\\
        &l = k,\dots,k{+}H{-}1.\label{eq:opt_con}
    \end{aligned}
\end{equation}
Applying this stepwise for $l = k,\dots,k{+}H{-}1$ yields a receding-horizon control law consistent with the $H$-step cost.
\end{theorem}

\vspace{-1em}
\begin{proof}
At each step $l$, the instantaneous stage cost can be written as
$\tilde{J}_i^{l}
= \sum_{j \in \mathcal{S}_{i,\mathrm{loc}}^{l}}
    \pi_{lj}^* \|A_i \mathbf{x}_i^{\,l} + B_i u_i^{\,l} - q_{j}^{\,l}\|^2
  + (u_i^{\,l})^\top R u_i^{\,l}$.
The first term quantifies the transport discrepancy between the propagated agent position 
and the nearby samples selected at step $l$, while the second term penalizes control energy to prevent excessive maneuvers.  
Here, $\pi_{lj}^*$ represents the portion of mass transported from the agent’s next predicted position 
to each sample $q_j^{\,l}$, establishing a direct coupling between optimal transport and control.

Applying the weighted-distance decomposition yields
$
\sum_{j}$ $\pi_{lj}^* \|A_i \mathbf{x}_i^{\,l} + B_i u_i^{\,l} - q_{j}^{\,l}\|^2
= \gamma_l \|A_i \mathbf{x}_i^{\,l} + B_i u_i^{\,l} - q_{i,c}^{\,l}\|^2
  + \sum_{j} \pi_{lj}^* \|q_{j}^{\,l} - q_{i,c}^{\,l}\|^2$,
where $\gamma_l = \sum_{j \in \mathcal{S}_{i,\mathrm{loc}}^{l}} \pi_{lj}^*$ is the total transported mass.  
The second term is independent of $u_i^{\,l}$ and thus does not affect the minimizer.  
Minimizing $\tilde{J}_i^{l}$ is therefore equivalent to minimizing
\[
\hat{J}_i^{l}
= \gamma_l \|A_i \mathbf{x}_i^{\,l} + B_i u_i^{\,l} - q_{i,c}^{\,l}\|^2
  + (u_i^{\,l})^\top R u_i^{\,l}.
\]
This formulation reveals a quadratic trade-off between tracking accuracy (through centroid alignment) 
and actuation cost (through the control penalty).

Taking the gradient with respect to $u_i^{\,l}$ and setting it to zero gives
\[
\nabla_{u_i^{\,l}}\hat{J}_i^{\,l}
= 2 \gamma_l B_i^\top (A_i \mathbf{x}_i^{\,l} + B_i u_i^{\,l} - q_{i,c}^{\,l})
  + 2 R u_i^{\,l} = 0,
\]
which leads to
$
(R + \gamma_l B_i^\top B_i) u_i^{\,l}
= \gamma_l B_i^\top (q_{i,c}^{\,l} - A_i \mathbf{x}_i^{\,l})$.
Solving for $u_i^{\,l}$ yields the analytic form in \eqref{eq:opt_con}.  
Since the horizon cost is additive and the centroid is updated at each step, 
applying the one-step optimizer sequentially provides a stagewise implementation 
of the finite-horizon control law.
\end{proof}

% -------------------------------------------------------------------
% Stage B 
% -------------------------------------------------------------------

\subsection{Stage B: Coverage-Tracking Weight Update}

After each motion step, each agent updates its local sample weights
to reflect the accumulated coverage of the surrounding region.  
This operation does not physically transport mass but serves as a 
\textit{coverage-tracking step} that records local exploration.
Consequently, samples in frequently visited areas lose influence,
while those in rarely visited regions retain higher importance.
The total empirical mass is strictly preserved,
ensuring consistency with the D$^2$OC conservation property.

Let \(d_{ij}^{k} = \|p_i^{k} - q_{j}^k\|\) denote the distance between agent~\(i\)
and sample~\(q_{j}^k\).
Samples closer to the agent are assigned higher removal priority as
\begin{equation}
r_j = \frac{d_{ij}^{k}}{1 + \beta \big((\sigma_j^{k})^2 + (\sigma_{j,\mathrm{virtual}}^{k})^2\big)},
\label{eq:priority r_j}
\end{equation}
where $\beta>0$ balances spatial proximity and uncertainty.  
Here, $\sigma_j^{k}$ is the intrinsic sensing variance,
and $\sigma_{j,\mathrm{virtual}}^{k}$ is the AI-adapted exploration variance.

Let $\mathcal{Q}_i^k \subseteq \mathcal{S}_i^k$ denote the subset of samples
selected for weight reduction in ascending order of~\eqref{eq:priority r_j},
such that their cumulative decrease satisfies
\begin{equation}
\textstyle
\sum_{j\in\mathcal{Q}_i^k} (w_j^{k} - w_j^{k+1}) = \frac{1}{M}.
\label{eq:weight_update}
\end{equation}
This update quantifies cumulative coverage progress 
while preserving total mass and preventing numerical drift.

% -------------------------------------------------------------------
% Stage C: Communication
% -------------------------------------------------------------------

\subsection{Stage C: Communication and Merging}

At each iteration, agents exchange their locally maintained sample sets with nearby peers within the communication range $r_c$.  
The neighbor set of agent~$i$ at step~$k$ is defined as
\vspace{-.9em}
\begin{equation}
\mathcal{N}_i^{k} = \{ j \mid \| p_i^{k} - p_j^{k} \| \le r_c \},
\end{equation}
and information exchange occurs only among such neighboring pairs.

Through local communication, each agent gathers the sample sets of its neighbors 
$\{\mathcal{S}_j^k \mid j \in \mathcal{N}_i^k\}$ 
and merges them with its own set $\mathcal{S}_i^k$.
Spatially redundant samples are unified based on positional proximity, and the merged set is downsampled to maintain a fixed size $N$ using a farthest-point selection strategy \cite{li2022adjustable}.
This guarantees that each agent preserves a compact yet diverse representation of the collective sampling history without unbounded growth in memory.

After merging, the sample weights are projected back to the feasible set of nonnegative measures to preserve mass consistency:
\vspace{-1em}
\begin{equation}
\textstyle
w_{j}^{k+1} \ge 0, \quad 
\sum_{j} w_{j}^{k+1} = \sum_{j} w_{j}^{k}.
\end{equation}
This procedure guarantees that each agent maintains a bounded yet informative local representation of the environment.

\medskip
\begin{remark}[Role of Stage C in Decentralized Consensus]
Stage C functions as the decentralized consensus layer of the framework.  
Through limited-range exchanges, agents reconcile their local sample representations 
to maintain globally consistent mapping and coverage information.  
This enables coherent control updates across the network even under partial communication and asynchronous operation.
\end{remark}

%=================================================
\section{Dual-MLP Inference for Mean–Variance and Adaptive Uncertainty Estimation}\label{sec:AI}

To guide decentralized multi-agent exploration, two neural networks (MLPs) estimate statistical properties of candidate samples. 
The first refines the locally measured mean and variance of the environmental field, while the second adaptively updates a \textit{virtual} variance reflecting how long a region has been unvisited. 
Together, these modules provide complementary exploitation and exploration cues without altering the core D$^2$OC control framework.

\subsection{Dual-MLP Architecture}

The framework employs two independent MLPs.

\noindent \textbf{MLP-MeanVar:} predicts local mean $\mu_j^{k}$ and intrinsic variance $(\sigma_j^{k})^2$ for samples $q_{j}^k \in \mathcal{S}_i^k$.
Inputs include the sample coordinates $q_{j}^k$, agent state $p_i^{k}$ (and optionally velocity), historical observations, and neighborhood statistics.
The network consists of two hidden layers with 64 neurons each and employs Rectified Linear Unit activations,
$\text{ReLU}(x)=\max(0,x)$.
It outputs a two-dimensional vector $y_j^{k} = [\mu_j^{k}, \log (\sigma_j^{k})^2]$,
where representing the variance in logarithmic form ensures numerical stability and positivity.
The predicted statistics are then substituted into the importance score \eqref{eq:phi}, which is used in Stage~A for centroid computation.
This network operates in inference mode only, with no online training.

\noindent \textbf{MLP-AdaptiveStd:} updates the virtual variance $\sigma_{j,\mathrm{virtual}}^{k}$ through online learning to promote exploration.
Its input concatenates the sample feature $z_j^{k}$ and the current $\sigma_{j,\mathrm{virtual}}^{k}$, and the output is the increment $\Delta\sigma_{j,\mathrm{virtual}}^{k}$.
It shares the same layer structure as MLP-MeanVar but has a single output neuron.
This network is trained online prior to Stage~B, so that long-unvisited samples receive larger virtual variance and thus reduced removal priority in \eqref{eq:priority r_j}.

\subsection{Integration with D$^2$OC}

\noindent \textbf{Sample Scoring and Weight Redistribution:}
Both Stage~A and Stage~B use the same importance score $\phi_j^{k}$.
In Stage~A, it determines the weighted centroid guiding agent motion toward informative or uncertain regions.
In Stage~B, the same score appears in the removal priority in \eqref{eq:priority r_j} so that regions with high total uncertainty retain mass longer, promoting balanced coverage.
This unified formulation maintains consistent exploration–exploitation trade-offs while preserving local mass conservation.

\subsection{Online Learning Procedure for Dual MLPs}

Both MLPs operate over the entire sample set $\mathcal{S}_i^k$ maintained by each agent.
The first network (MLP-MeanVar) serves as a fixed regressor providing $(\mu_j^k, (\sigma_j^k)^2)$ for each sample, which can be pre-trained or replaced by user-defined parameters depending on the experimental setup.
The second network (MLP-AdaptiveStd) is updated online before Stage~B using visitation history and recent measurements to estimate a virtual uncertainty term for long-unvisited or newly generated samples.
This additional term is combined with the baseline variance in both Stage~A’s centroid computation and Stage~B’s redistribution priority.

\noindent\textbf{Forward Computation:}
For the $L$-layer MLP-MeanVar, the forward pass is defined as
\begin{align*}
h^{(0)} &= x_j, \quad y_j = W^{(L)} h^{(L-1)} + b^{(L)} =
\begin{bmatrix}
\mu_j^{k},\, \log (\sigma_j^{k})^2
\end{bmatrix}^{\!\top},\\
h^{(l)} &= \psi\!\big(W^{(l)} h^{(l-1)} + b^{(l)}\big), \quad l = 1,\dots,L{-}1,
\end{align*}
where $x_j$ includes the spatial coordinates $q_j^k$, agent state, and historical sensor measurements.
The activation $\psi(\cdot)$ is ReLU, $\psi(x)=\max(0,x)$.
The logarithmic form ensures numerical stability and positivity of the predicted variance.

Similarly, for MLP-AdaptiveStd,
\begin{align*}
z^{(0)}_j &= [q_{j}^k;\, \sigma_{j,\mathrm{virtual}}^{k}], \quad 
\Delta \sigma_j = V^{(L)} z^{(L-1)} + c^{(L)},\\
z^{(l)}_j &= \psi\!\big(V^{(l)} z^{(l-1)} + c^{(l)}\big), \quad l = 1,\dots,L{-}1,
\end{align*}
where the input concatenates the sample feature with its current virtual variance.
The output $\Delta\sigma_j$ is applied to update the virtual variance as
\[
\sigma_{j,\mathrm{virtual}}^{k+1}
= \sigma_{j,\mathrm{virtual}}^{k} + \Delta\sigma_j.
\]

\noindent\textbf{Loss Function:}
Since MLP-MeanVar operates only in inference mode, no loss is computed for it.
MLP-AdaptiveStd, however, minimizes a simple regularization loss:
\begin{equation*}
\textstyle
\mathcal{L}_\text{AdaptiveStd}
= \frac{1}{|\mathcal{S}_i^k|}
\sum_{q_{j}^k \in \mathcal{S}_i^k}
\big(\Delta \sigma_{j,\mathrm{virtual}} - \Delta \sigma_{\mathrm{target}}\big)^2.
\end{equation*}
This loss enforces stability by minimizing unnecessary changes in the virtual standard deviation.
The target value $\Delta \sigma_{\mathrm{target}} = 0$ is chosen as the nominal per-step behavior, i.e., if a sample is being sensed regularly, its virtual uncertainty should not keep growing.
In other words, no further inflation of $\sigma_{j,\mathrm{virtual}}$ is desired in the well-observed (steady-state) case. Only long-unvisited samples, detected through the online update rule, produce nonzero increments, whereas frequently visited ones remain nearly unchanged.

\noindent\textbf{Parameter Update:}
After computing $\mathcal{L}_\text{AdaptiveStd}$ at step $k$,
the network parameters are updated using a gradient-descent rule:
\[
\theta_\text{AdaptiveStd}^{k+1}
= \theta_\text{AdaptiveStd}^{k}
- \eta\,\nabla_{\theta_\text{AdaptiveStd}}
\mathcal{L}_\text{AdaptiveStd}^{k},
\]
where $\eta>0$ denotes the learning rate.
The updated parameters are then used in the subsequent control step \(k{+}1\) to predict
\[
\Delta \sigma_{j,\mathrm{virtual}}^{k+1}
= f_\text{AdaptiveStd}(q_{j}^k, p_i^{k+1}; \theta_\text{AdaptiveStd}^{k+1}).
\]

\noindent\textbf{Algorithmic Summary:}
At each control step \(k\), the following procedure is executed:
\begin{enumerate}
    \item Measure local field values and compute empirical $\hat{\mu}_j$, $\hat{\sigma}_j^2$.
    \item Predict $(\mu_j^k, (\sigma_j^k)^2)$ via the fixed MLP-MeanVar.
    \item Compute $\Delta\sigma_{j,\mathrm{virtual}}$ via MLP-AdaptiveStd and update 
    $\sigma_{j,\mathrm{virtual}}^{k+1} = \sigma_{j,\mathrm{virtual}}^{k} + \Delta\sigma_{j,\mathrm{virtual}}$.
    \item Compute the loss $\mathcal{L}_\text{AdaptiveStd}$ and update parameters using backpropagation.
    \item Proceed to Stages~A and~B using updated virtual variances and predicted means/variances.
\end{enumerate}

\medskip
\begin{remark}[Dual-Network Roles and Advantages]
The dual-network design separates static inference from adaptive learning.
MLP-MeanVar provides stable baseline statistics for environmental reconstruction,
while MLP-Adaptive Std continuously refines exploration bias through online updates.
This separation maintains a clear distinction between intrinsic field uncertainty and exploration-induced virtual uncertainty, achieving a stable exploration–exploitation balance within the decentralized D$^2$OC framework.
\end{remark}

%=================================================
\section{Convergence Analysis}

This section establishes that the iterative execution of the sensor-driven renewal 
and Stages~A--C in the proposed D$^2$OC framework drives the empirical distribution 
$\hat{\rho}$ toward the ground-truth density $\rho_{\mathrm{GT}}$ in the Wasserstein sense. 
At each control step, agents locally minimize the quadratic cost~\eqref{eq:cost_func} 
by moving toward their weighted centroids~\eqref{eq:centroid}, 
while the renewal process dynamically adds or removes samples based on the sensed field intensity. 
This continual update of sample weights refines $\hat{\rho}$ around $\rho_{\mathrm{GT}}$ 
as the environment is progressively explored. 
The overall evolution of $\hat{\rho}$ thus combines 
optimal local motion, bounded auxiliary adjustment, 
and intermittent consensus among neighboring agents. 
Under these coupled dynamics, the following convergence property holds.

\medskip
\begin{theorem}[Band-Limited Convergence of AI-Augmented D$^2$OC]
\label{thm:convergence}
Consider $N_{\mathrm{a}}$ agents evolving under the discrete-time LTI dynamics~\eqref{eq:LTI}.
Agent $i$ at step $k$ maintains an empirical, nonstationary measure
\[
\textstyle
\hat{\rho}_i^{k} = \sum_j w_{j}^{k}\,\delta(x - q_{j}^{k}), 
\qquad
\hat{\rho}^{k} = \frac{1}{N_{\mathrm{a}}}\sum_{i=1}^{N_{\mathrm{a}}}\hat{\rho}_i^{k}.
\]
The measure $\hat{\rho}^{k}$ is updated at each step by the D$^2$OC control loop, consisting of a sensor-driven renewal (pre-update) and the control–communication stages (A–C).

Assume the following conditions hold:
\begin{enumerate}
    \item \textbf{Pre-update (Dynamic sample generation and removal).}
    Before each control update, every agent renews its local sample set using new sensor feedback.
    High-value regions generate new samples, while those whose field intensity falls below a distance-dependent threshold are removed.
    This birth–death renewal perturbs the empirical distribution by a uniformly bounded amount,
    referred to as the \emph{renewal disturbance}, satisfying $\varepsilon_{\mathrm{renew}}>0$.

    \item \textbf{Stage~A (Local optimal control).}
    Each agent applies the analytic optimal control~\eqref{eq:opt_con} derived in Theorem~\ref{thm:opt_con},
    producing a contractive update in the Wasserstein metric.% with some $0<\gamma_A<1$.

    \item \textbf{Stage~B (Coverage-Tracking redistribution).}
    After the motion update, each agent redistributes its sample weights to
    represent how much of the surrounding region has been explored.
    The redistribution priority follows \eqref{eq:priority r_j},
    where $\beta>0$ balances spatial proximity and combined uncertainty.  
    The total empirical mass of each agent remains exactly conserved,
    i.e., $\sum_{j\in\mathcal{Q}_i^k} w_{j}^{k+1} = \sum_{j\in\mathcal{Q}_i^k} w_{j}^{k} = 1/M$.
    This step reassigns existing weights without creating new mass and therefore introduces no additional disturbance to convergence.

    \item \textbf{Stage~C (Intermittent consensus).}
    Communication occurs over a jointly connected graph,
    guaranteeing that local empirical measures asymptotically agree in expectation.
\end{enumerate}

Then, there exist constants $0<\underline{\mathcal{W}}\le\overline{\mathcal{W}}$ such that the nonstationary empirical distribution satisfies
\[
\underline{\mathcal{W}}
\le
\liminf_{k\to\infty} \mathcal{W}_2(\hat{\rho}^k, \rho_{\mathrm{GT}})
\le
\limsup_{k\to\infty} \mathcal{W}_2(\hat{\rho}^k, \rho_{\mathrm{GT}})
\le
\overline{\mathcal{W}},
\]
where $\overline{\mathcal{W}} = \tfrac{\varepsilon_{\mathrm{renew}}}{1 - \gamma_A}$ and $\underline{\mathcal{W}} > 0$,  
with $\varepsilon_{\mathrm{renew}} > 0$ denoting the uniformly bounded disturbance introduced by the sensor-driven renewal step and $0 < \gamma_A < 1$ representing the contraction rate of the local control mapping.

\end{theorem}

\vspace{-0.5em}
\begin{proof}
Let $\mathcal{T}_R$, $\mathcal{T}_A$, $\mathcal{T}_B$, and $\mathcal{T}_C$ denote,
respectively, the update mappings corresponding to
(i) sensor-driven renewal,
(ii) local optimal control,
(iii) coverage-tracking redistribution, and
(iv) intermittent consensus.

\textit{1) Renewal step (sample birth-and-death).}
Each agent renews its sample set based on sensor feedback:
new samples are generated near high-GT regions and low-value ones are discarded according to a distance-dependent threshold.
Because sensing is recurrent and both sensor noise and update radii are bounded,
the induced perturbation on the empirical distribution is also uniformly bounded:
\[
\mathcal{W}_2(\mathcal{T}_R(\hat{\rho}^{k}), \rho_{\mathrm{GT}}) 
\le 
\mathcal{W}_2(\hat{\rho}^{k}, \rho_{\mathrm{GT}}) + \varepsilon_{\mathrm{renew}},
\quad \varepsilon_{\mathrm{renew}}>0.
\]

\textit{2) Stage~A: contraction through local optimal control.}
At step $k$, agent $i$ solves the $H$-step cost~\eqref{eq:cost_func}, 
where the weighted centroid~\eqref{eq:centroid} is computed from the updated samples in $\hat{\rho}_i^{k}$.
The analytic optimal control~\eqref{eq:opt_con} from Theorem~\ref{thm:opt_con} yields
\[
\mathbf{x}_i^{l+1} = F_l \mathbf{x}_i^{l} + G_l q_{i,c}^{l},
\quad
F_l = A_i - B_i (R + \gamma_l B_i^\top B_i)^{-1} \gamma_l B_i^\top A_i.
\]
If $F_l$ is Schur-stable ($R\!\succ\!0$, moderate $\gamma_l$),
the mapping $\mathcal{T}_A$ is contractive in $\mathcal{W}_2$:
\[
\mathcal{W}_2(\mathcal{T}_A(\hat{\rho}^{k}), \rho_{\mathrm{GT}})
\le
\gamma_A \, \mathcal{W}_2(\hat{\rho}^{k}, \rho_{\mathrm{GT}}),
\quad 0<\gamma_A<1.
\]

\textit{3) Stage~B: coverage-tracking redistribution.}
After motion, each agent redistributes its sample weights according to the 
priority $r_j^{k}$ while keeping $\sum_{j\in\mathcal{Q}_i^k} w_{j}^{k+1} = \sum_{j\in\mathcal{Q}_i^k} w_{j}^{k} = 1/M$.
Since this process merely redistributes fixed mass over existing samples,
it is nonexpansive in $\mathcal{W}_2$ and introduces no additional disturbance:
\[
\mathcal{W}_2(\mathcal{T}_B(\hat{\rho}^{k}), \rho_{\mathrm{GT}})
\le
\mathcal{W}_2(\hat{\rho}^{k}, \rho_{\mathrm{GT}}).
\]

\textit{4) Stage~C: intermittent consensus.}
Agents communicate only within range $r_c$, forming a time-varying graph $\mathcal{G}^k$ with averaging matrix $A^k=[a_{il}^k]$.
Under joint connectivity, the sequence $\{A^k\}$ is ergodic~\cite{moreau2005stability}, ensuring
\[
\mathbb{E}\bigl[\mathcal{W}_2(\hat{\rho}_i^{k}, \hat{\rho}_\ell^{k})\bigr]\to 0,
\quad \forall i,\ell,
\]
thus $\mathcal{T}_C$ is mean-nonexpansive in $\mathcal{W}_2$.

\textit{5) Composite inequality.}
Combining the mappings yields
\[
\mathcal{W}_2(\hat{\rho}^{k+1}, \rho_{\mathrm{GT}})
\le
\gamma_A \, \mathcal{W}_2(\hat{\rho}^{k}, \rho_{\mathrm{GT}})
+ \varepsilon_{\mathrm{renew}}.
\]
Iterating gives
\[
\textstyle
\limsup_{k\to\infty}\mathcal{W}_2(\hat{\rho}^{k}, \rho_{\mathrm{GT}})
\le
\frac{\varepsilon_{\mathrm{renew}}}{1-\gamma_A}
=:\overline{\mathcal{W}}.
\]

\textit{6) Positive error floor.}
Since $\hat{\rho}^k$ is a finite-sample approximation of continuous $\rho_{\mathrm{GT}}$, a nonzero mismatch remains:
\[
\liminf_{k\to\infty}\mathcal{W}_2(\hat{\rho}^k,\rho_{\mathrm{GT}})\ge\underline{\mathcal{W}}>0,
\]
so $\hat{\rho}^k$ converges within $[\underline{\mathcal{W}},\overline{\mathcal{W}}]$, limited by finite-sample and renewal errors.
\end{proof}

\begin{remark}[On the necessity of dual MLPs]
The convergence in Theorem~\ref{thm:convergence} holds for the AI-augmented D$^2$OC,
where the dual MLPs ensure that Stage~A remains contractive in practice.
Without these AI modules, the control update may lose its contraction property since agents may stagnate near locally dense regions, and
the Wasserstein distance stops decreasing, violating the assumed $\gamma_A<1$ condition.
Hence, the \textit{MLP-MeanVar} and \textit{MLP-AdaptiveStd} are essential not for tightening
the theoretical bound itself but for realizing the contraction behavior that the theorem presumes.
\end{remark}

%=================================================
\section{Simulations}
To evaluate the proposed AI-augmented D$^2$OC framework, we performed simulations under a representative abstraction of a \emph{landfill methane plume detection} scenario.  
Detecting small and transient methane plumes is inherently challenging since satellite-based sensing often fails to distinguish weak emission signals from background variations due to limited spatial resolution and atmospheric scattering~\cite{rouet2024automatic}.  
In contrast, UAV- or ground-based sensing offers higher accuracy but is constrained by flight time, sensing range, and intermittent communication~\cite{shaw2021methods}.  
The simulation environment captures these essential characteristics, multiple localized emission sources, limited sensing range, and communication constraints, allowing systematic evaluation of the algorithm’s adaptability and convergence across diverse environmental mapping conditions.

\subsection{Simulation Setup}
The simulation parameters used in this study are summarized in Table~\ref{tab:sim_params_final}. These parameters define a $200\times200$~m$^2$ domain, AI training configuration, and quadrotor dynamic limits, representing a realistic mid-scale environment suitable for validation of decentralized mapping performance.

In Fig.~\ref{fig:sim_traj}(a), the simulated landfill gas mapping environment is illustrated. The red-shaded region represents the ground-truth (GT) plume distribution, which is unknown to the multi-agent system, while five linearized quadrotors are deployed as robotic platforms. The blue dots indicate the initial sample distribution, representing the system’s prior estimate of the plume distribution. Even if such prior information is unavailable, the proposed method can flexibly operate using a uniform reference distribution as a baseline. 
As shown in Fig. \ref{fig:sim_traj}, the reference samples are deviated from the GT to provide an initial estimation error on purpose.

\begin{table}[h!]
\centering
\caption{Simulation Parameters}
\label{tab:sim_params_final}
\scriptsize
\setlength{\tabcolsep}{3.5pt}   % column spacing compressed for 2-column layout
\renewcommand{\arraystretch}{0.9} % row spacing compressed
\begin{tabular}{l l}
\hline
Parameter & Value (Unit) \\
\hline
Domain Size & $200\times200$ (m$^2$)\\
Sample Size ($N_{\mathrm{s}}$)  & 300 (-)\\
Sampling Interval ($\Delta t$) & 0.2 (s) \\
Number of Agents ($N_{\mathrm{a}}$) & 5 (-) \\
Operation Time ($T_{\text{op}}$) & 600 (s) = 10 (mins)\\
Total Control Steps ($M$) & 3000 (-) \\
Sensing/Comm. Range ($r_\mathrm{s}$/$r_\mathrm{c}$) & 10 (m)/20 (m) \\
Sample Create/Drop Thresholds & $1{\times}10^{-5}$ (-)/$2{\times}10^{-6}$ (-)\\ 
Uncertainty Weights ($c_1$/$c_2$) & 1.0 (-)/1.0 (-) \\
Trade-off weight ($\beta$) & 1.0 (-) \\
Horizon ($H$) / Input Weight ($R$) & 5 (-) / $0.01 I_2$ (-) \\
Max Vel/Max Tilt/Angular Rate & 5 (m/s)/$10$ ($^\circ$)/$30$ ($^\circ$/s) \\
Hidden Layers/Neurons per Layer & 2 (-)/64 (-) \\
Learning Rate ($\eta$) & $1{\times}10^{-3}$ (-) \\
MLP Update Interval ($T_{\text{MLP}}$) & 10 (step) \\
\hline
\end{tabular}
\vspace{-.1in}
\end{table}

\begin{figure}[t!]
    \subfloat[]{
    \includegraphics[width=0.5\linewidth]{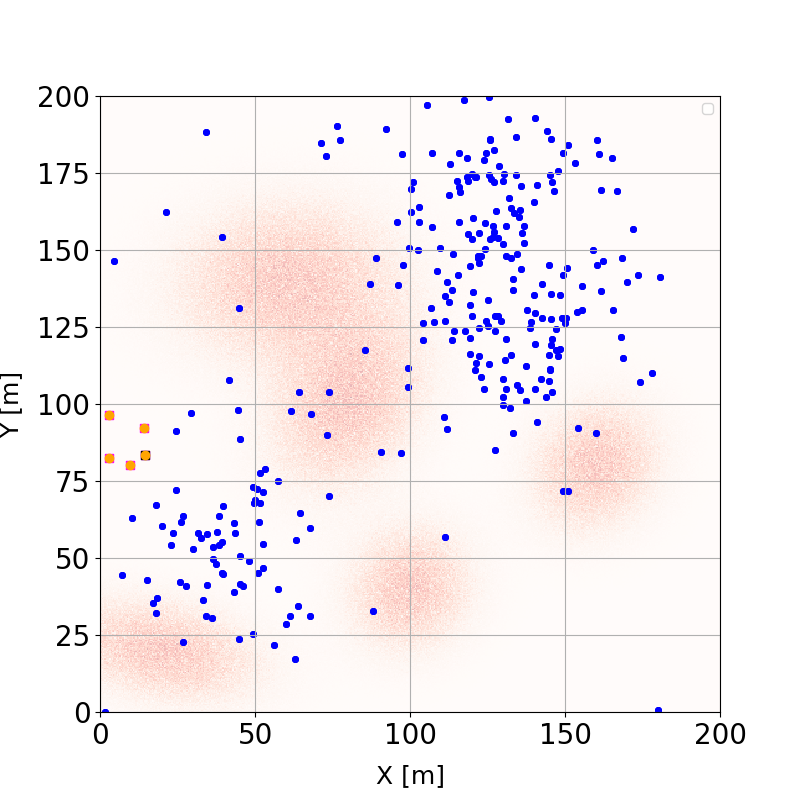}}
    \subfloat[]{
    \includegraphics[width=0.5\linewidth]{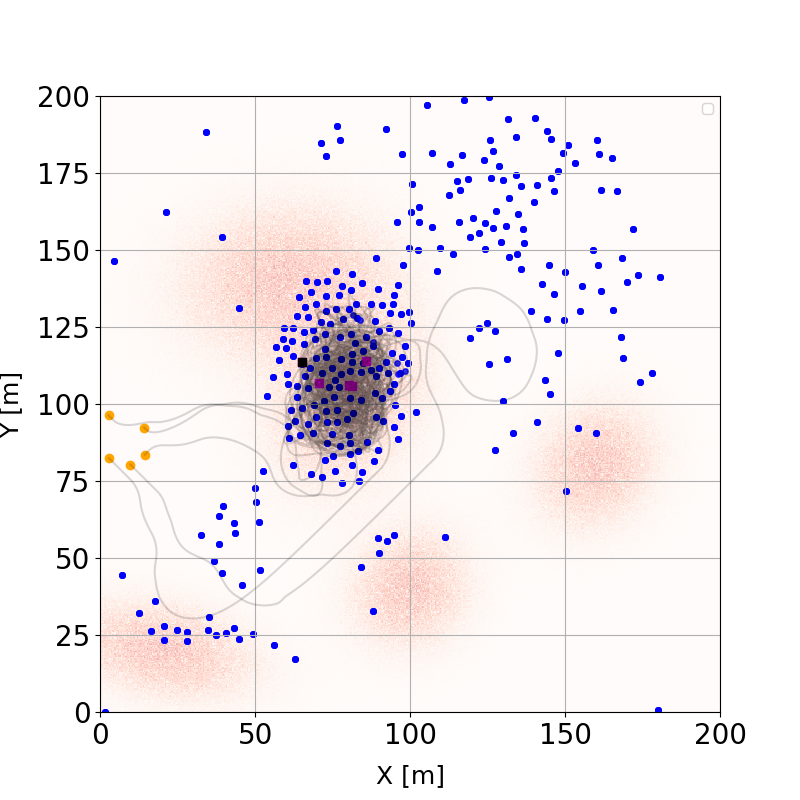}}\\
    \subfloat[]{
    \includegraphics[width=0.5\linewidth]{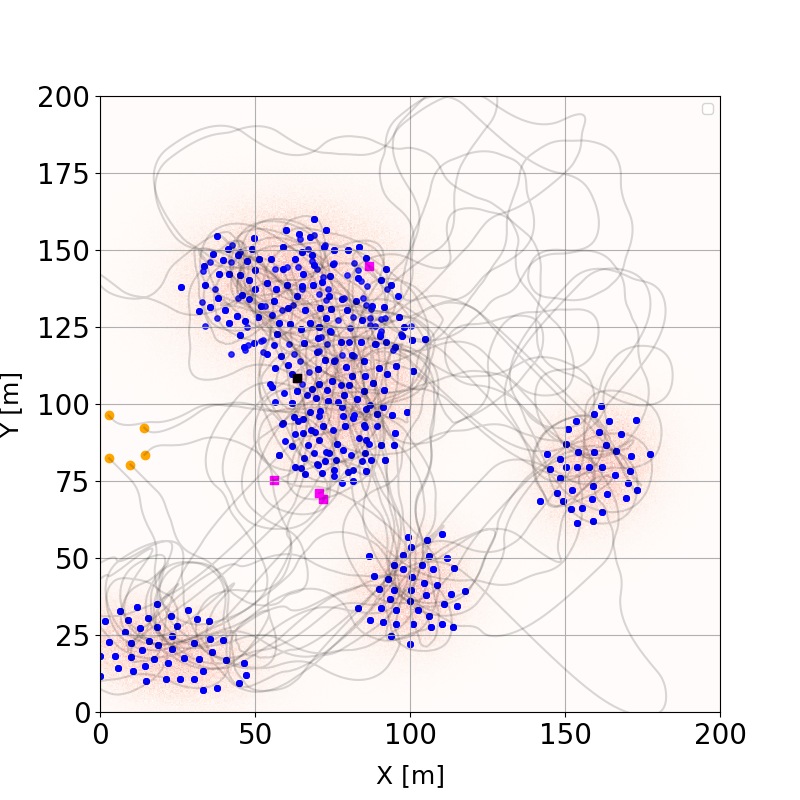}}
    \subfloat[]{
    \includegraphics[width=0.46\linewidth]{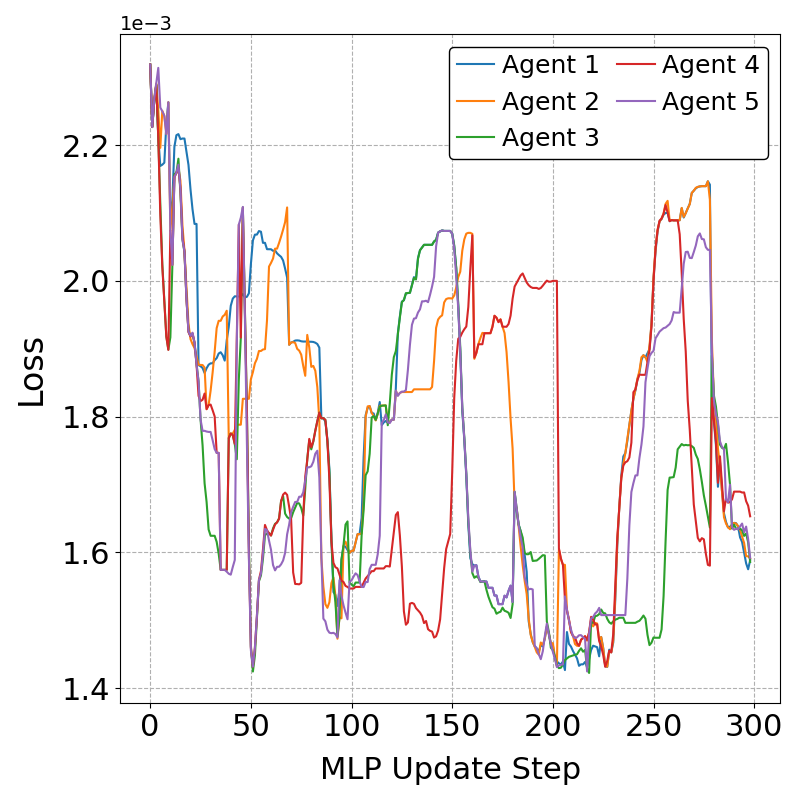}}
    \caption{Simulation results:
    (a) Initial reference samples (blue), GT distribution (shaded red), and initial positions of five agents;  
    (b) Updated sample distribution after mission completion without AI augmentation;
    (c) Results of AI-augmented D$^2$OC showing improved coverage efficiency compared to (b);  
    (d) Time history of the loss function during the adaptive learning process.}\label{fig:sim_traj}
\end{figure}

\subsection{Result Analysis}

Figs.~\ref{fig:sim_traj}(b,c) illustrate the multi-agent trajectories under two control strategies: (b) D$^2$OC without MLP and (c) AI-augmented D$^2$OC with MLP. The initial and final positions of the five agents are shown by orange circles and magenta squares, respectively.
The optimal control inputs are obtained using Theorem~\ref{thm:opt_con} with horizon $H=5$.
In simulation, input bounds on velocity and attitude rates are applied only as post-processing limits to maintain physical plausibility and numerical stability, while the theoretical control law itself remains unconstrained. A visual comparison reveals that the AI-augmented approach reconstructs the plume distribution more accurately. Without AI, agents tend to oversample near high-density regions and become trapped in local minima due to the absence of an escape mechanism. The proposed adaptive variance mechanism mitigates this by increasing the virtual uncertainty of unvisited samples, encouraging exploration of underrepresented regions. Consequently, all plume sources are successfully detected, and the final sample distribution aligns closely with the ground truth.

The learning behavior of the MLP-AdaptiveStd is shown in Fig.~\ref{fig:sim_traj}(d). The network is updated every ten control steps over a total of 3000 control steps. The loss fluctuates within a bounded range of approximately $[1.4\!-\!2.4]\times10^{-3}$, reflecting the dynamic nature of online adaptation. Since MLP-AdaptiveStd continuously updates to align with the evolving virtual variance fields rather than a static target, the regularization loss does not decrease monotonically. Intermittent communication and local sample merging cause distribution shifts, yet the bounded oscillation of the loss indicates that the adaptation remains stable and well-regulated under decentralized operation.

The Wasserstein distance comparison in Fig.~\ref{fig:w_dist} quantitatively evaluates the distributional alignment between the reconstructed and ground-truth densities.
The flat orange line corresponds to the static initial reference, while the blue and red curves represent D$^2$OC without and with MLP, respectively. 
As shown in Fig.~\ref{fig:w_dist} and consistent with Theorem~\ref{thm:convergence}, the Wasserstein distance in the AI-augmented D$^2$OC case shows an overall decreasing trend with minor fluctuations, whereas the non-MLP baseline remains nearly stagnant.
The simulation results confirm this theoretical behavior, showing that the AI-augmented version consistently maintains a smaller Wasserstein distance, indicating more stable convergence and higher reconstruction fidelity under sensing and communication constraints.
\begin{figure}[t!]
    \centering
    \includegraphics[width=0.68\linewidth]{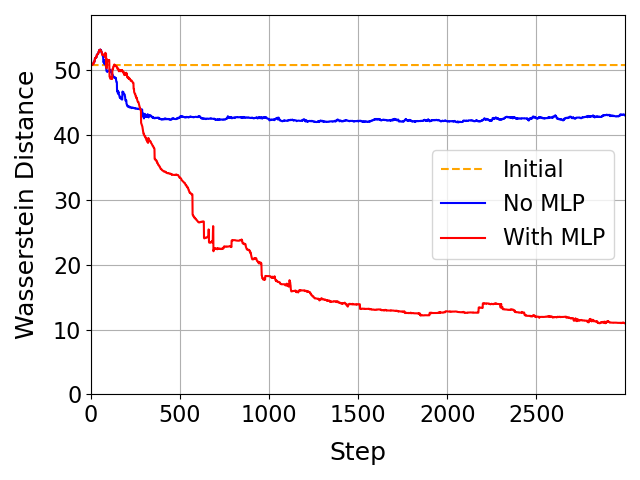}
    \vspace{-1em}
\caption{Wasserstein distance comparison with GT:
initial (orange); no MLP (blue); AI-augmented (red).}
    \label{fig:w_dist}
\end{figure}

%=================================================
\section{Conclusion}
This paper presented an AI-augmented D$^2$OC framework for decentralized multi-agent mapping under sensing and communication constraints.  
By integrating dual-MLP inference into D$^2$OC, each agent adaptively prioritized exploration based on predicted mean and uncertainty, enabling efficient reconstruction of complex spatial distributions without centralized coordination.  
The proposed framework preserves the theoretical convergence guarantees of D$^2$OC under the Wasserstein metric while improving robustness to measurement noise, dynamic sample renewal, and intermittent communication.  
Simulation results demonstrated that the proposed method achieved lower steady-state Wasserstein distance and higher map fidelity than conventional decentralized approaches, confirming improved global consistency in spatial coverage.  
Future work will extend the framework to dynamic environments, heterogeneous agent systems, and real-world UAV experiments for large-scale environmental monitoring and disaster response.

% \begin{ack}
% Place acknowledgments here.
% \end{ack}

% \bibliographystyle{ifacconf}
\bibliography{references}

\end{document}